\documentclass[11pt]{article}

\usepackage{amsmath,amssymb,amsfonts,amsthm} 
\usepackage{mathrsfs,latexsym} 
\usepackage{mathtools}

\newtheorem{theorem}{Theorem}[section]

\newtheorem{lemma}[theorem]{Lemma}

\numberwithin{equation}{section}

\DeclareMathAlphabet{\pazocal}{OMS}{zplm}{m}{n}

\usepackage[mathscr]{eucal} 

\usepackage{color}
\usepackage{cite}
\usepackage{epsfig}
\usepackage{graphicx}

\usepackage{bm}
\usepackage{cite}

\def\CC{\mathbb{C}}
\def\NN{\mathbb{N}}
\def\RR{\mathbb{R}}

\newcommand{\cB}{{\mathcal{B}}}

\newcommand{\cF}{{\mathcal{F}}}

\def\bk{\mbox{\boldmath $k$}}
\def\sbk{{\mbox{\scriptsize \boldmath $k$}}}
\newcommand{\bp}{{\mbox{\boldmath $p$}}}
\newcommand{\sbp}{{\mbox{\scriptsize \boldmath $p$}}}
\newcommand{\bx}{{\mbox{\boldmath $x$}}}
\newcommand{\sbx}{{\mbox{\scriptsize \boldmath $x$}}}
\newcommand{\by}{{\mbox{\boldmath $y$}}}
\newcommand{\sby}{{\mbox{\scriptsize \boldmath $y$}}}

\def\bO{\mbox{\boldmath $O$}}

\def\bP{\mbox{\boldmath $P$}}
\newcommand{\sbP}{{\mbox{\scriptsize \boldmath $P$}}}

\def\eg{{\it e.g.\ }}
\def\viz{{\it viz.}}
\def\ie{{\it i.e.\ }}

\def\be{\begin{equation}}
\def\ee{\end{equation}}    

\begin{document}

\title{Proper condensates and off-diagonal long range order}
\author{Detlev Buchholz \\[2mm]
Mathematisches Institut, Universit\"at G\"ottingen, \\
\small Bunsenstr.\ 3-5, 37073 G\"ottingen, Germany \\ 
\small detlev.buchholz@mathematik.uni-goettingen.de \\[5pt]
}
\date{}

\maketitle

\noindent \textbf{Abstract.} Within the framework of the algebra of
canonical commutation relations in Euclidean space,
a long range order between particles 
in bounded regions is established in states 
with a sufficiently large particle number. It    
occurs whenever homogeneous proper
(infinite) condensates form locally in
the states in the limit of infinite densities.    
The condensates are described by eigenstates
of the momentum operator, covering also those cases, where 
they are streaming with a constant velocity. The arguments
given are model independent and lead to a new criterion
for the occurrence of condensates. 
It makes use of a novel approach to the identification of
condensates, based on a characterization of regular and singular
wave functions.

Off-diagonal long range order


\section{Introduction}
\label{sec1}
\setcounter{equation}{0} 

An important tool for the experimental detection of
Bose-Einstein condensates are interference measurements
on trapped gases. They are sensitive to the appearance 
of coherent configurations of particles which all have 
the same momentum, irrespective of their distance,
cf.\ \cite{KeDuSt}. On the
theoretical side, this phenomenon manifests itself in
the absence of decay properties of correlation functions,
generally referred to as off-diagonal long-range order (ODLRO). It
leads to peak values of the respective Fourier
transforms for these coincident momenta, cf.\
for example \cite{PaWa}.

\medskip
In this note we take a fresh look at this topic,
starting from a novel characterization of condensates \cite{Bu1}.
Instead of following the Onsager-Penrose approach, where one
characterizes condensates by the largest eigenvalues and
corresponding eigenfunctions of one-particle density matrices,
we focus on the spaces of \textit{regular} wave functions, 
which remain finitely occupied in the limit of infinite particle
numbers. Given any open, bounded region $\bO \subset \RR^d$, a  
proper (infinite) condensate in that region, 
appearing in the limit of infinite particle densities,  
is identified with the orthogonal complement of the resulting 
regular functions with support in $\bO$. The functions
in this orthogonal complement are said to be \textit{singular}.
This approach has the advantage that condensates, appearing in a finite
system, can be identified in a clear-cut manner, which does not depend
on its global shape or specific number of particles. 

\medskip
Making use of these notions, we will establish in 
Fock-states on the algebra of canonical commutation relations
the appearance of ODLRO in   bounded regions if the particle density 
is sufficiently high. This happens whenever homogeneous proper 
condensates are formed in the limit of infinite densities. It turns out that
these 
proper condensates can be described by eigenstates of the
momentum operator, possibly with a momentum different from zero.
Hence, our results cover also the situation where the
condensate propagates with a constant velocity.  

\medskip
The notation and concepts used in this paper, in particular the notion 
of proper condensates and their manifestations, are briefly recalled in Sec.~2.
Our main result, concerning the appearance of off-diagonal
long range order, is established in Sec.\ 3.
The paper concludes with a brief summary and a new criterion for
the occurrence of homogeneous proper condensates. In an appendix,
we provide examples, showing that such condensates
can appear under conditions which are weaker than the
Onsager-Penrose criterion for Bose-Einstein condensation. 

\section{Proper condensates}
\label{sec2}
\setcounter{equation}{0} 

Let $\cF = \bigoplus_n \cF_n$ be the bosonic Fock space, \viz, the direct
sum of $n$-particle spaces $\cF_n$ that are defined by the
$n$-fold symmetrized tensor products of the single particle
space $\cF_1 \doteq L^2(\RR^d)$, \mbox{$n \in \NN_0$}. We interpret the elements
\mbox{$f,g \in L^2(\RR^d)$} as single particle wave functions
on the $d$-dimensional
position space $\RR^d$. Their canonical scalar product is denoted 
by~$\langle f, g \rangle$. On $\cF$, annihilation operators
$a(f)$ and creation operators $a^*(g)$ are densely defined.
They are antilinear and linear in their entries, respectively, and satisfy
the commutation relations
\be
[a(f), a^*(g)] = \langle f, g \rangle \, 1 \, , \quad f,g \in L^2(\RR^d) \, ,
\ee
all other commutators being equal to $0$. 

\medskip
We will
consider sequences of states $\omega$, which can be represented by density
matrices $\rho$ on $\cF$,
\be
\omega(A) \doteq \text{Tr} \, \rho A \, , \quad A \in \cB(\cF) \, .
\ee
Since these sequences have limits which are no longer representable
in this manner, it is meaningful to restrict the states to suitable
subalgebras of the algebra~$\cB(\cF)$, \ie the algebra of bounded
operators on $\cF$. A standard choice is the Weyl algebra, another 
convenient choice is the resolvent algebra, invented in
\cite{BuGr} and used for the analysis of condensates in
\cite{Bu1, BaBu}. We need not delve into these 
issues here and can restrict our attention to the
so-called single particle density matrices, given by
\be
f,g \mapsto \omega(a^*(f)a(g)) \, , \quad f,g \in L^2(\RR^d) \, .
\ee
It is assumed in the following that
$\omega(a^*(f)a(f)) < \infty$ for all $f  \in L^2(\RR^d)$, which is
the case if the state $\omega$
contains a finite (mean) number of particles.

\medskip
We fix in the following an open, bounded region $\bO \subset \RR^d$, which
is thought of as being of macroscopic size, \ie big compared to typical
microscopic 
length scales. The subspace of functions with support in that region
is denoted by $L^2(\bO)$. Given any sequence of states, 
we define a corresponding regular subspace of $L^2(\bO)$ as
follows.

\medskip \noindent
\textbf{Definition I:}  Let $\omega_\sigma$, $\sigma > 0$, be a sequence
of states with properties described above. The corresponding
\textit{regular} subspace $R(\bO) \subset L^2(\bO)$ consists of all
functions $f \in  L^2(\bO)$ satisfying
\be
\limsup_{\sigma \rightarrow \infty} \, \omega_\sigma(a^*(f) a(f)) < \infty \, . 
\ee
Its complement $L^2(\bO) \backslash R(\bO)$ consists of \textit{singular}
wave functions, which are infinitely occupied in the limit. 

\medskip
There are many reasons why single particle wave functions
with support in a bounded region can be infinitely occupied
in the limit states, whence are singular.
For example, this may be due to high energy effects, as is the 
case for equilibrium states approaching infinite temperatures. There
all wave functions are non-regular in the limit. Yet these cases are of
no interest in the present context. We 
rely here on a more specific characterization of sequences of states
that eventually exhibit a proper condensate. In these states
there appear, besides clouds of
particles with a regular wave function, 
increasing numbers of particles which all 
occupy the same singular wave function $s$, cf.\ \cite{Bu1,BaBu}.

\medskip \noindent
\textbf{Definition II:} Let $\omega_\sigma$, $\sigma > 0$, be a sequence
of states with properties as in Definition I. The limit
of the sequence contains a proper (infinite) condensate in $\bO$
whenever $R(\bO)$ is closed
and has a one-dimensional orthogonal complement
in $L^2(\bO)$, consisting of the ray spanned by some singular function~$s$. 
This function characterizes 
a condensate that appears for sufficiently large values of $\sigma$. 

\medskip \noindent
\textbf{Remark:} In \cite{Bu1} the possibility was also discussed that 
the orthogonal complement of $R(\bO)$
has a finite dimension, different from one. 
We restrict our attention here to the preceding important case,  
known to appear in many models, cf.\ for example \cite{BaBu}. 
Relevant examples are also recalled in Sec.~4.

\medskip
Let $\omega_\sigma$, $\sigma > 0$, be a sequence
of states with the properties described in Definition II 
and let $s \in L^2(\bO)$ be the (up to a phase unique) normalized function
in the orthogonal complement of $R(\bO)$. Putting
\be
\sigma \mapsto n_C(\sigma) \doteq \omega_\sigma(a^*(s) a(s)) \, ,
\ee
this sequence is unbounded in the limit of large $\sigma$.
For the corresponding divergent  subsequences, it
defines the number of particles in $\bO$, forming a
proper condensate in the limit.

\begin{lemma}
  Let $\omega_\sigma$, $\sigma > 0$, be a sequence of states with properties
  specified in Definition II. The renormalized
  one-particle density matrices  
  \be
f,g \mapsto n_C(\sigma)^{-1} \omega_\sigma(a^*(f)a(g)) \,,
\quad f,g \in L^2(\bO) \, ,
\ee
converge for suitable subsequences of $\sigma$ in norm to the
one-dimensional projection onto the (normalized) singular
wave function~$s$, 
\be \label{e.2.7} 
\lim_\sigma  n_C(\sigma)^{-1} \omega_\sigma(a^*(f)a(g))    = 
\langle s, f \rangle \langle g, s \rangle   \, ,
\quad f,g \in L^2(\bO) \, .
\ee
\end{lemma}
\begin{proof}
  Since $R(\bO)$ is closed, it follows from the uniform boundedness
  principle~\cite{Yo} that there is some constant $n_R$ such that,
uniformly with regard to~$\sigma$, 
\be
\omega_\sigma(a^*(f)a(f)) \leq n_R
\, \| f \|^2 \, , \quad f \in R(\bO) \, . 
\ee
Now, decomposing
\be
f = f_\perp + \langle s , f \rangle  s \, ,
\quad f \in L^2(\bO) \, ,
\ee
whence $f_\perp \in R(\bO)$, one obtains 
\begin{align} \label{e.2.10}
| \omega_\sigma(a^*(f)a(f)) & - n_C(\sigma) | \langle f, s \rangle |^2 | \nonumber \\
& \leq n_R \, \| f_\perp \|^2 +
2(n_R n_C(\sigma))^{1/2} \, \| f_\perp \| \, | \langle s, f \rangle |
\nonumber \\
& \leq \big(n_R + 2(n_R n_C(\sigma))^{1/2} \big) \, \| f \|^2 \, .
\end{align}
Thus, picking any subsequence of $\sigma$ for which $n_C(\sigma)$
diverges, the corresponding sequence of renormalized one-particle density
matrices converges, as stated. \end{proof}  

\medskip
The preceding lemma shows that  
the existence of proper condensates in $\bO$ 
manifests itself in a clear-cut manner already 
in the approximating sequence of Fock-states.
We emphasize that this does not necessarily
imply that the number of particles
in these states with wave function $s$
is of macroscopic order, 
\ie it need not be proportional
to the total (expected) number of particles in $\bO$. It merely
must exceed for large $\sigma$ the maximal possible number $n_R$ of
particles occupying some regular wave function. 
As is shown in the subsequent section,
this feature already implies the existence of ODLRO in locally
homogeneous states. 

\section{Off-diagonal long range order}
\label{sec3}
\setcounter{equation}{0} 

The appearance of ODLRO in states containing 
condensates is commonly proven by proceeding to
the thermodynamic limit of Gibbs-von Neumann ensembles.
This is done either by unfolding a trapping
potential and adjusting the number of particles
or by considering systems with constant density 
in infinitely growing boxes. 
We will show here that ODLRO can also be established
in fixed bounded regions for sequences of Fock-states,
containing an increasing number of particles
which form a homogeneous proper condensate in the limit.
As we shall discuss in Sec.\ 4, this result covers the preceding
cases.

\medskip 
Proceeding to the details, we recall that the region $\bO$,
which was fixed above, is open and bounded. We will consider
open subregions
$\bO_0 \subset \bO$ whose closure is contained in $\bO$
and write in this case $\bO_0 \Subset \bO$.
Thus, for sufficiently small translations $\bx \in \RR^d$, 
one has also $\bO_0 + \bx \Subset \bO$. Denoting by $\bP$
the momentum operator on $L^2(\RR^d)$, this implies 
$e^{i \sbx \sbP} L^2(\bO_0) \subset L^2(\bO)$ for such
translations. 

\medskip \noindent
\textbf{Definition III:} \ Let $\omega_\sigma$, $\sigma > 0$, be a
sequence of states as in Lemma~2.1. This sequence 
describes a \textit{homogeneous} proper condensate in $\bO$
in the limit if 
\be
\lim_\sigma \, n_C(\sigma)^{-1} 
\omega_\sigma(a^*(e^{i \sbx \sbP} \! f) a(e^{i \sbx \sbP} \! g))
= \langle s, f \rangle \langle g, s \rangle 
\ee
for all $f,g \in L^2(\bO_0)$, $\bO_0 \Subset \bO$, and 
translations $\bx \in \RR^d$ such that  $\bO_0 + \bx \Subset \bO$.

\medskip 
It is apparent from Definition I  that the regular functions
$R(\bO_0)$ with support in $\bO_0 \Subset \bO$ are
contained in $R(\bO)$. As a matter of fact, the assignment
$\bO_0 \mapsto R(\bO_0)$ defines a net on
the subsets $\bO_0 \Subset \bO$. It follows from Definition~III
that nets resulting from the corresponding states are also stable under small
translations. For, the orthogonal
complement of $R(\bO)$ in~$L^2(\bO)$ coincides with the ray of $s$,
which implies for sufficiently small $\bx$ 
\be  \label{e.3.2}
| \langle s, e^{i \sbx \sbP} f \rangle |^2 =
| \langle s, f \rangle |^2 = 0 \, , \quad f \in R(\bO_0) \, .
\ee
Thus, $ e^{i \sbx \sbP}   R(\bO_0) \subset \{ s \}^\perp \bigcap L^2(\bO_1)
= R(\bO_1)$, provided $\bO_0 + \bx \subset \bO_1 \Subset \bO$.
With this information, we can determine now the possible 
form of the singular function $s$. 

\begin{lemma} \label{l.3.1}
  Let $\omega_\sigma$, $\sigma > 0$, be a sequence of states that 
  describes a homogeneous proper condensate in $\bO$
  in the limit. The corresponding
  singular wave function $s$ has the form
  \be
  \bx \mapsto s(\bx) =
  \begin{cases} \label{e.3.3}
    |\bO|^{-1/2} \, e^{i \sbx \sbp}  &  \text{if} \quad \bx \in \bO \\
    0   & \text{if} \quad  \bx \in \RR^d \backslash \bO \, ,  
  \end{cases}
    \ee 
    where $\bp \in \RR^d$ and $|\bO|$ is the volume of $\bO$. 
  \end{lemma}  

\begin{proof}
  Given $\bO_0 \Subset \bO$, the orthogonal complement
  of $R(\bO_0)$ in $L^2(\bO_0)$ is the ray of $s_0$,
  which coincides with
  the normalized restriction $s \upharpoonright \bO_0$.
  (Since $\bO_0$ is open, it follows after a moments reflection
  that this restriction is different from $0$.)
  Thus  $f - \langle s_0, f \rangle s_0 \in R(\bO_0)$
  for any $f \in L^2(\bO_0)$. We choose now regions
  $\bO_0 \Subset \bO_1 \Subset \bO$ with corresponding
  singular functions $s_0$ and $s_1$. Then, 
  for sufficiently small $\bx, \by$ such that
  $\bO_0 + \bx \Subset \bO_1$ and $\bO_1 + \by \Subset \bO$, we have 
  \begin{align}
  \langle s, e^{i(\sbx + \sby) \sbP} f \rangle
  & = \langle s, e^{i \sby \sbP} s_1 \rangle \nonumber 
  \langle s_1, e^{i \sbx \sbP} f \rangle   \nonumber \\
  & = \frac{\langle s, e^{i \sby \sbP} s_1 \rangle}{\langle s, s_1 \rangle}
 \, \langle s, e^{i \sbx \sbP} f \rangle \, , \quad f \in L^2(\bO_0) \, .
  \end{align}
  Since the matrix elements of the unitary translation operators are
  continuous, it follows from this equality
  that there is some $\bk \in \CC^d$ such that for small $\bx$ 
  \be
  \bx \mapsto  \langle s, e^{i \sbx \sbP} f \rangle
  =  \langle s, f \rangle \, e^{\sbx \sbk } \, ,
   \quad f \in L^2(\bO_0) \, .
   \ee
   Equation \eqref{e.3.2} then implies that $\bk = i \bp$ for
   some $\bp \in \RR^d$. Since the region $\bO_0 \Subset \bO$
   can be arbitrarily chosen, this completes the proof. 
\end{proof}  

The proof that sequences of states, describing a homogeneous
proper condensate in $\bO$ in the limit, exhibit ODLRO is now straightforward. 
To fix ideas, we choose a length $L_0$ of 
microscopic (\eg atomic) size and consider balls 
\mbox{$\bO_0 \Subset \bO$}, centered at
the origin of $\RR^d$, with diameters $L_0$ and $L \gg L_0$, respectively.
Picking any normalized function  $f \in L^2(\bO_0)$, we make use
of Lemma~\ref{l.3.1}, which yields 
$|\langle s, e^{i \sbx \sbP} f \rangle| \leq (|\bO_0|/|\bO|)^{1/2}$ for
$|\bx| < (L-L_0)/2$. 
Relation~\eqref{e.2.10} then implies 
for condensate densities $n_C(\sigma)/|\bO| > n_R/|\bO_0|$ that 
\be \label{e.3.6}
\omega_\sigma(a^*(e^{i \sbx \sbP} f) a(e^{i \sby \sbP} f))
= (n_C(\sigma)/|\bO|) \, e^{i(\sbx - \sby) \sbp} \, (2 \pi)^s
|\widetilde{f}(\bp)|^2 \, ,
\ee
disregarding contributions of order
$(n_R n_C(\sigma) |\bO_0| / |\bO|)^{1/2}$; 
the tilde~$\, \widetilde{ } \,$ denotes Fouri\-er transforms.
Thus, for sufficiently large
condensate densities $n_C(\sigma)/|\bO|$, the correlations
between almost all particles in the state have
in leading order constant, non-vanishing amplitudes
at distances up to $(L-L_0)$. Hence ODLRO prevails
in the approximating Fock-states. 

\medskip 
Thinking of interference experiments, 
it is also of interest to determine the Fourier transforms of
the correlation functions. To this end we consider the
localized plane waves 
$\bx \mapsto e_\sbk(\bx)$, $\bk \in \RR^d$,  which are
defined as in equation~\eqref{e.3.3}, putting $\bp = \bk$. 
Relations \eqref{e.2.7} and
\eqref{e.2.10} then imply that for $n_C(\sigma) > n_R$
\be \label{e.3.7} 
\omega_\sigma(a^*(e_\sbk)a(e_\sbk)) = n_C(\sigma) \, 
s^2 (L |\bk - \bp|/2)^{-2s}
\Big( \int_0^{L |\sbk - \sbp|/2} \! dr \, r^{s-2} \sin(r) \Big)^2 \, ,
\ee
disregarding contributions of order $(n_R n_C(\sigma))^{1/2}$. 
Thus, for $\bk = \bp$, the corresponding number of particles
coincides with the number $n_C(\sigma)$ of
particles in the condensate. If 
$L |\bk - \bp| > 2$, a straightforward estimate shows that the
corresponding number is smaller than 
$4 n_C(\sigma) (L |\bk - \bp| )^{-2}$.
Given the magnitude of $L$, it follows that the momentum distribution of the
particles in $\bO$ has a pronounced peak at the momentum $\bp$ of
the  particles in the condensate. We summarize these results in the
following theorem. 

\begin{theorem}
  Let $\bO_0 \Subset \bO \subset \RR^d$ be
  concentric balls of diameter $L_0$ and $L$, respectively, and   
  let $\omega_\sigma$, $\sigma > 0$, be a sequence of Fock-states that 
  describes a homogeneous proper condensate in~$\bO$ in the
  limit. For values of $\sigma$ such that
  $n_C(\sigma)/|\bO| > n_R/|\bO_0| $, there appear in $\omega_\sigma$
  undamped correlations between particles localized in distant 
  balls $\bO_0 + \bx, \bO_0 + \by  \Subset \bO$ (ODLRO)
  at distances $|\bx - \by| < L - L_0$, cf.\ equation~\eqref{e.3.6}.
  If $n_C(\sigma) > n_R$, the momentum distribution of the 
  particles in~$\bO$ in the state $\omega_\sigma$
  has a peak at the joint momentum of the
  particles forming the condensate, cf.\ equation~\eqref{e.3.7}.
\end{theorem}  

We emphasize, that these characteristic properties
of condensates appear locally in Fock-states with a finite
particle number, provided the density of the
local condensate complies with the given constraints.

\section{Conclusions}
\label{sec4}
\setcounter{equation}{0} 

In the present article we have established the existence of long range
correlations in bosonic
systems of a limited number of particles in a bounded region. The only 
input used was the assumption that the systems can in principle be enlarged
in a manner that leads to the formation of homogeneous proper 
condensates. That is, there exists a closed subspace of regular wave
functions of co-dimension one, having support in the region, whose
occupation numbers remain finite, whereas their one-dimensional
orthogonal complement can be occupied by an unlimited
number of particles. 

\medskip
This input comprises the qualitative features of
condensation phenomena, found in experiments with trapped Bose gases. 
We did not need to assume that the systems are in
equilibrium or to specify any dynamics. It was also not necessary
to assume that the condensate wave functions are macroscopically
occupied, \ie that the corresponding number
of particles is of the same order of magnitude as the total
number of particles. It only matters that the number
of particles in the condensate substantially exceeds the maximal possible
number of particles occupying some regular wave function.
If this is given, the existence of the condensate manifests
itself in a pronounced peak of the momentum distribution of the
particles in the state, which is localized at the common momentum of the
particles in the condensate.

\medskip
Our arguments are based only on kinematic properties 
of systems of Bosons. Yet 
the question of whether proper condensates appear depends~of
course on the dynamics. Simple examples are equilibrium 
states of non-interacting
Bosons in a fixed box in
any number of dimensions $d$, where homogeneous proper condensates are formed in
the limit of infinite particle numbers. More interesting 
are systems of non-interacting Bosons, which are confined by
some smooth trapping potential. There one must  
simultaneously increase the number of particles and unfold 
the trapping potential. The resulting states are in general not
homogeneous, providing examples where the spatial translations are
spontaneously broken in the limit. Nevertheless, homogeneous proper
condensates appear for large particle numbers. They occupy increasing
neighborhoods of the minimum of the trapping potential, cf.\ for
example \cite{BaBu,Bu1}. 

\medskip
The proof that homogeneous proper condensates exist in interacting
systems is more difficult. Such systems are frequently analyzed by relying on
approximations of mean field type. There the existence of proper condensates 
can be extracted from the literature,
cf.\ for example \cite{Le} and references quoted there. A major challenge,
however, is a proof in case of genuine two-body interactions. 
In view of the present results, it
amounts to a comparison of the occupation numbers of the
largest two eigenvalues of localized one-particle density matrices;
it is not necessary to obtain control on the full spectrum. 
As a matter of fact, one 
can rely on the following criterion, where we
restrict our attention to the case  
of condensates having vanishing momentum,
such as in rotational invariant states.

\medskip \noindent
\textbf{Criterion:} Let $\omega_\sigma$, $\sigma > 0$, be a sequence of
Fock-states, let $\bO$ be an open bounded region, and let
${R}(\bO)$ be the (closed)
subspace of functions $f \in L^2(\bO)$ satisfying $\int \! d\bx \, f(\bx) = 0$.
The sequence describes in the limit a homogeneous proper 
condensate in $\bO$ with zero momentum if and only if 
\be
\limsup_\sigma \, \omega_\sigma(a^*(f) a(f)) < \infty \, ,
\quad f \in {R}(\bO) \, ,  
\ee
and, for some (hence any) function $s \in L^2(\bO)$ with 
$\int \! d\bx \, s(\bx) \neq 0$, one has 
\be
\limsup_\sigma \, \omega_\sigma(a^*(s) a(s)) = \infty \, .
\ee

\noindent
\textbf{Remark:} Making use of relation \eqref{e.2.10}, there 
are less stringent, quantitative versions of this criterion, which   
likewise entail the properties of the approximating Fock-states,
presented in Theorem~3.2. 

\medskip
To summarize, the concept of proper condensates
provides a meaningful picture of condensation 
and its features, even if phase transition points (\eg temperatures)
are not sharply defined. The concept also allows to analyze 
inhomogenities in the spatial structures of coexisting phases,
cf.~\cite{Bu1}.
Thus, our results provide a fresh look at
the longstanding problem to establish the 
existence of Bose-Einstein condensates in the presence
of genuine interactions. 

\bigskip  \noindent 
\textbf{\Large Acknowledgments}

\medskip \noindent
I gratefully acknowledge stimulating discussions with Jakob Yngvason
and the hospitality and financial support of the
Erwin Schr\"odinger Institute in Vienna. I am
also grateful to Dorothea Bahns and the Mathematics Institute
of the University of G\"ottingen for their continuing hospitality. 

\section*{Conflict of interest}
There are no relevant financial or non-financial
competing interests to disclose.

\section*{Data availability} 
Data sharing is not applicable to this article as no new
data were created or analyzed in this study.

\appendix
\section*{Appendix: Illustrative examples}
\setcounter{section}{1}
\setcounter{equation}{0}

In this appendix we present some examples of sequences of states 
that describe homogeneous proper condensates in the limit,
but do not comply with the Onsager-Penrose criterion for Bose-Einstein
condensation. These examples are merely of a mathematical nature.
But they show that our concept of proper condensation differs 
markedly from the Onsager-Penrose approach.
Even though our condition is less stringent, its implications
with regard to observable effects of condensation are quite similar. 
For example, there appear characteristic peaks in the momentum
distributions of particles in a gas containing such a condensate. 

\medskip
We fix in the following some open,  bounded region $\, \bO \subset \RR^s$.
Let \mbox{$e_k \in L^2(\bO)$}, $k \in \NN_0$, be an orthonormal basis,
where we choose for $e_0$ the constant function in $\bO$. 
Picking some $0 < \varepsilon < 1$, we define for
numbers $n > 1$ the quantities 
$n_C(n) \doteq  n^\varepsilon$ and 
$\varepsilon_n \doteq \ln((1 + n - n_c(n))/(n - n_C(n)) > 0$. It entails 
\be
n_c(n) + \sum_{k=1}^\infty e^{- \varepsilon_n k} =
n_c(n) + e^{- \varepsilon_n}(1 - e^{- \varepsilon_n})^{-1} = 
n \, . 
\ee
We then define for $n > 1$ a 
sequence of gauge invariant quasifree states $\omega_n$
on the algebra of canonical commutation relations, putting  
\be
\omega_n(a^*(e_k) a(e_l)) \doteq
\delta_{k,l} \,
\begin{cases}
  n_C(n) &  \text{if} \ \ k = 0 \\
  e^{- \varepsilon_n k} &  \text{if} \ \ k \geq 1 \, .
\end{cases}  
\ee
These states may be arbitrarily extended to the full algebra, for example
as product states on the given region $\bO$ and
its complement $\RR^s \backslash \bO$. It follows from this definition that 
\be
\omega_n(a^*(e_0)a(e_0)) = n_C(n) \, , \quad
\sum_{k=0}^\infty \omega_n(a^*(e_k)a(e_k)) = n \, .
\ee
The (mean) number of particles
in this sequence with wave function in
the condensate state $e_0$ increases
with $n$ like $n^\varepsilon$, whereas
the maximal occupation number of particles in $\bO$ with wave function 
in the orthogonal complement of $e_0$ is bounded by $1$. Since  
the total (mean) number of particles in $\bO$ in the states is
equal to $n$, the condensate is not macroscopically occupied. But
the sequence describes a homogeneous proper condensate
in $\bO$ in the limit with all of its consequences, discussed
in the main text. 

\medskip 
For given orthonormal basis, one can also define
Hamiltonians on $L^2(\bO)$ with corresponding eigenstates and
arbitrary discrete spectrum. The states $\omega_n$ are stationary
under the action of the corresponding dynamics. Given
$n > 1$ and a temperature $T$,
there are also Hamiltonians $H_n$ on $L^2(\bO)$ such that the
states $\omega_n$ satisfy the KMS condition for the corresponding
dynamics and given temperature. They act on the orthonormal basis
$e_k$, $k \in \NN_0$,  according to
\be
H_n \, e_k = T e_k 
\begin{cases}
  \ln(1 + n_C(n)^{-1}) & \text{if} \ \ k = 0 \\
  \ln (1 + e^{\, \varepsilon_n k}) &  \text{if} \ \ k \geq 1 \, .
\end{cases}
\ee
One can extend these Hamiltonians to wave functions with
arbitrary support by adding to them operators which commute and 
act trivially on the wave functions with support in
$\bO$. There exist also more refined examples of
this kind.

\end{document}